\documentclass[10pt]{article}

\usepackage[T1]{fontenc}
\usepackage[sc]{mathpazo}
\usepackage{amsmath}
\usepackage{amssymb}
\usepackage{enumerate}
\usepackage{amsthm}% theorems and proofs
\usepackage{amsfonts,mathrsfs}
\usepackage[style]{fncychap}
\usepackage{graphicx, color} % For including graphics N.B. pdftex graphics driver
\usepackage{geometry} % allows easy specifying of the page layout
\usepackage{eepic}
\usepackage{ifthen}

\usepackage{subfigure}
\usepackage{multirow}
\newboolean{ElectronicVersion}
\setboolean{ElectronicVersion}{true} % CHANGE THIS AS REQUIRED

\usepackage[unicode=true,pdfusetitle, bookmarks=true,bookmarksnumbered=false,bookmarksopen=false, breaklinks=false,pdfborder={0 0 0},backref=false,colorlinks=false] {hyperref}
\hypersetup{
colorlinks,linkcolor=myurlcolor,citecolor=myurlcolor,urlcolor=myurlcolor}
\definecolor{myurlcolor}{rgb}{0,0,0.7}% Color highlighted for numbers in citations and theorems

\usepackage{hyperref}
\hypersetup{pdfpagemode=UseNone}

\newcommand{\tinyspace}{\mspace{1mu}}

\newcommand{\proj}[1]{| #1\rangle\!\langle #1 |}

\newcommand{\abs}[1]{\left\lvert\tinyspace #1 \tinyspace\right\rvert}

\newcommand{\norm}[1]{\left\lVert\tinyspace #1 \tinyspace\right\rVert}

\renewcommand{\det}{\operatorname{det}}

\newcommand{\setft}[1]{\mathrm{#1}}

\newcommand{\density}[1]{\setft{D}\left(#1\right)}

\def\complex{\mathbb{C}}
\def\real{\mathbb{R}}
\def\natural{\mathbb{N}}

\def\I{\mathbb{1}}

\def \dif {\mathrm{d}}
\def \diag {\mathrm{diag}}
\def \re {\mathrm{Re}}

\newenvironment{mylist}[1]{\begin{list}{}{
    \setlength{\leftmargin}{#1}
    \setlength{\rightmargin}{0mm}
    \setlength{\labelsep}{2mm}
    \setlength{\labelwidth}{8mm}
    \setlength{\itemsep}{0mm}}}
    {\end{list}}

\def\ot{\otimes}

\newcommand{\Pa}[1]{\left(#1\right)}

\newcommand{\Br}[1]{\left[#1\right]}
\newcommand{\set}[1]{\{#1\}}
\newcommand{\Set}[1]{\left\{#1\right\}}

\newcommand{\bra}[1]{\langle#1|}

\newcommand{\ket}[1]{|#1\rangle}

\DeclareMathOperator{\trace}{Tr}

\newcommand{\Ptr}[2]{\trace_{#1}\Pa{#2}}

\newcommand{\Tr}[1]{\Ptr{}{#1}}

\def\cI{\mathcal{I}}

\def\cX{\mathcal{X}}

\def\rU{\mathrm{U}}

\def\sC{\mathscr{C}}

\def\L{\textsf{L}}

\newcommand{\cmp}{Comm. Math. Phys.~}
\newcommand{\ieee}{IEEE. Trans. Inf. Theory~}
\newcommand{\jmp}{J. Math. Phys.~}
\newcommand{\jpa}{J. Phys. A~}

\newcommand{\natphy}{Nature Phys.~}
\newcommand{\natcomm}{Nature Commun.~}

\newcommand{\njp}{New. J. Phys.~}

\newcommand{\prl}{Phys. Rev. Lett.~}
\newcommand{\pra}{Phys. Rev. A~}
\newcommand{\pre}{Phys. Rev. E~}
\newcommand{\prx}{Phys. Rev. X~}

\newcommand{\rmop}{Rev. Mod. Phys.~}

\newtheorem{thrm}{Theorem}[section]
\newtheorem{lem}[thrm]{Lemma}
\newtheorem{prop}[thrm]{Proposition}
\newtheorem{cor}[thrm]{Corollary}

\theoremstyle{definition}

\numberwithin{equation}{section}

\newcounter{questionnumber}

\begin{document}

%============================================================================================================%
\title{\large Average subentropy, coherence and entanglement of random mixed quantum states}
%============================================================================================================%

\author{Lin Zhang\footnote{E-mail: godyalin@163.com; linyz@zju.edu.cn}\\
  {\it\small Institute of Mathematics, Hangzhou Dianzi University, Hangzhou 310018, PR~China}\\
  Uttam Singh\footnote{E-mail: uttamsingh@hri.res.in},\, Arun K. Pati\footnote{E-mail: akpati@hri.res.in}\\
  {\it \small Harish-Chandra Research Institute, Allahabad, 211019, India}}

\date{}
\maketitle

\begin{abstract}
Compact expressions for the average subentropy and coherence are
obtained for random mixed states that are generated via various
probability measures. Surprisingly, our results show that the
average subentropy of random mixed states approaches to the maximum
value of the subentropy which is attained for the maximally mixed
state as we increase the dimension. In the special case of the
random mixed states sampled from the induced measure via partial
tracing of random bipartite pure states, we establish the typicality
of the relative entropy of coherence for random mixed states
invoking the concentration of measure phenomenon. Our results also
indicate that mixed quantum states are less useful compared to pure
quantum states in higher dimension when we extract quantum coherence
as a resource. This is because of the fact that average coherence of
random mixed states is bounded uniformly, however, the average
coherence of random pure states increases with the increasing
dimension. As an important application, we establish the typicality
of relative entropy of entanglement and distillable entanglement for
a specific class of random bipartite mixed states. In particular,
most of the random states in this specific class have relative
entropy of entanglement and distillable entanglement equal to some
fixed number (to within an arbitrary small error), thereby hugely
reducing the complexity of computation of these entanglement
measures for this specific class of mixed states.\\~\\
\textbf{AMS 2000 subject classifications.} 81P40, 62P35\\
\textbf{Keywords and phrases.} Quantum coherence, von Neumann
entropy, quantum subentropy, random quantum state, Selberg integral

\end{abstract}

%===========================================================================%
\section{Introduction}
%===========================================================================%

Miniaturization \cite{Linden2010} and technological advancements to
handle and control systems at smaller and smaller scales necessitate
the deeper understanding of concepts such as quantum coherence,
entanglement and correlations
\cite{Aspuru13,Horodecki2013,Skrzypczyk2014,Narasimhachar2015,Brandao2015,Rudolph214,Rudolph114,Plenio2008,Aspuru2009,Lloyd2011,Huelga13,Levi14}.
Two inequivalent resource theories of coherence has been proposed
\cite{Gour2008, Marvian14, Baumgratz2014} realizing the importance
of the coherence as a resource in various physical situations.
Recently, it has been proved that the coherence of a random pure
state sampled from the uniform Haar measure is generic  for higher
dimensional systems, i.e., most of the random pure states have
almost the same amount of coherence \cite{UttamB2015}. The
importance of this result and the similar results for entanglement
of random bipartite pure states cannot be overemphasized. The
average entanglement of random bipartite pure states, which is
facilitated by the calculation of average entropy of the marginals
of the random bipartite pure states
\cite{Page1993,Foong1994,Jorge1995, Sen1996}, is proved typical
\cite{Hayden2006}. This has resulted in various interesting
consequences in quantum information theory
\cite{Hayden2006,Hamma2012,Dahlsten2014,ZhangB2015,Nakata2015}, in
the context of black holes \cite{PageB1993} and in particular, in
explaining the {\it equal a priori probability} postulate of
statistical physics \cite{Popescu2006,Goldstein2006}. But as we
approach towards the realistic implementations of quantum
technology, mixed states are encountered naturally due to the
interaction between the system of interest and the external world.
Therefore, consideration of average entanglement and coherence
content of random mixed states is of great importance in realistic
scenarios. However, to the best of our knowledge, there is no known
result on the average coherence of random mixed states.

Here, we aim at finding the average relative entropy of coherence of
random mixed states sampled from various induced measures including
the one obtained via the partial tracing of the Haar distributed
random bipartite pure states. We first find the exact expression for
the average subentropy of random mixed states sampled from induced
probability measures and use it to find the average relative entropy
of coherence of random mixed states. We note that the subentropy is
a nonlinear function of state and therefore, it is expected that the
average subentropy of a random mixed state should not be equal to
the subentropy of the average state (the maximally mixed state).
Surprisingly, we find that the average subentropy of a random mixed
state approaches exponentially fast towards the maximum value of the
subentropy, which is achieved for the maximally mixed state
\cite{JozsaB1994}. As one of the applications of our results, we
note that the average subentropy may also serve as the state
independent quality factor for ensembles of states to be used for
estimating accessible information. Interestingly, we find that the
average coherence of random mixed states, just like the average
coherence of random pure states, shows the concentration phenomenon.
This means that the relative  entropies of coherence of most of the
random mixed states are equal to some fixed number (within an
arbitrarily small error) for larger Hilbert space dimensions. It is
well known that the exact computation of the most of the
entanglement measures for bipartite mixed states in higher
dimensions is almost impossible \cite{HorodeckiRMP09}. However,
using our results, we compute the average relative entropy of
entanglement and distillable entanglement for a specific class of
random bipartite mixed states and show their typicality for larger
Hilbert space dimensions. It means that for almost all random states
of this specific class, both the measures of entanglement are equal
to a fixed number (that we calculate) within an arbitrarily small
error, reducing hugely the computational complexity of both the
measures for this specific class of bipartite mixed states. This is
a very important practical application of the results obtained in
this paper.

%=====================================================================================================%
\section{Quantum coherence and induced measures on the space of mixed states}\label{sec:rand-state}
%=====================================================================================================%

\subsection{Quantum coherence}

Various coherence monotones, that serve as the faithful measures of
coherence \cite{Baumgratz2014, Alex15, Winter2015, UttamA2015}, are
proposed based on the resource theory of coherence
\cite{Baumgratz2014}. These monotones include the $l_1$ norm of
coherence, relative entropy of coherence \cite{Baumgratz2014} and
the geometric measure of coherence based on entanglement
\cite{Alex15}. In this work, unless stated otherwise, by coherence
we mean the relative entropy of coherence throughout the paper. The
relative entropy of coherence of a quantum state $\rho$, acting on
an $m$-dimensional Hilbert space, is defined as
\cite{Baumgratz2014}: $\sC_r(\rho) := S(\Pi(\rho)) - S(\rho)$, where
$\Pi(\rho)=\sum^m_{j=1}\proj{j}\rho\proj{j}$ for a fixed basis
$\set{\ket{j}:j=1,\ldots,m}$. $S(\rho) = -\Tr{\rho\ln\rho}$ is the
von Neumann entropy of $\rho$. All the logarithms that appear in the
paper are with respect to natural base.

\subsection{Induced measures on the space of mixed states}

Unlike on the set of pure states, it is known that there exist
several inequivalent measures on the set of density matrices,
$\density{\complex^m}$ (the set of trace one nonnegative $m\times m$
matrices). By the spectral decomposition theorem for Hermitian
matrices, any density matrix $\rho$ can be diagonalized by a unitary
$U$. It seems natural to assume that the distributions of
eigenvalues and eigenvectors of $\rho$ are independent, implying
$\mu$ to be product measure $\nu\times\mu_{\mathrm{Haar}}$, where
the measure $\mu_{\mathrm{Haar}}$ is the unique Haar measure on the
unitary group and measure $\nu$ defines the distribution of
eigenvalues but there is no unique choice for it
\cite{Zyczkowski2001, Zyczkowski2003}.

The induced measures on the $(m^2-1)$-dimensional space
$\density{\complex^m}$ can be obtained by partial tracing the
purifications $\ket{\Psi}$ in the larger composite Hilbert space of
dimension $mn$ and choosing the purified states according to the
unique measure on it. Following Ref. \cite{Zyczkowski2001}, the
joint density $P_{m(n)}(\Lambda)$ of eigenvalues
$\Lambda=\set{\lambda_1,\ldots,\lambda_m}$ of $\rho$, obtained via
partial tracing, is given by
\begin{eqnarray}
P_{m(n)}(\Lambda) =
C_{m(n)}K_1(\Lambda)\prod^m_{j=1}\lambda^{n-m}_j\theta(\lambda_j),
\end{eqnarray}
where the theta function $\theta(\lambda_j)$ ensures that $\rho$ is
positive definite, $C_{m(n)}$ is the normalization constant and
$K_1(\Lambda)$ is given by
\begin{eqnarray}\label{eq:k-lambda}
K_\gamma(\Lambda) =
\delta\Pa{1-\sum^m_{j=1}\lambda_j}\abs{\Delta(\lambda)}^{2\gamma},
\end{eqnarray}
for $\gamma=1$ with $\Delta(\lambda)=\prod_{1\leqslant i<j\leqslant
m}(\lambda_i-\lambda_j)$. See Refs.
\cite{Zyczkowski2001,Zyczkowski2003} for a good exposition of
induced measures on the set of density matrices.

Now we introduce the family of integrals $\cI_m(\alpha,\gamma)$ that will be a key to our work, where
\begin{eqnarray}\label{eq:alpha-beta}
\cI_m(\alpha,\gamma)&:=& \int^\infty_0\cdots\int^\infty_0 K_\gamma(\Lambda)\prod^m_{j=1}\lambda^{\alpha-1}_j\dif\lambda_j\nonumber\\
&=&b_m(\alpha,\gamma)\prod^m_{j=1}\frac{\Gamma\Pa{\alpha+\gamma(j-1)}\Gamma\Pa{1+\gamma
j}}{\Gamma\Pa{1+\gamma}},
\end{eqnarray}
with $\alpha,\gamma>0$, $\Gamma(z):=\int^\infty_0 t^{z-1}e^{-t}\dif
t$ is the Gamma function, defined for $\re(z)>0$ and
$b_m(\alpha,\gamma)=\set{\Gamma\Pa{\alpha m+\gamma m(m-1)}}^{-1}$.
The value of above family of integrals can be obtained using
Selberg's integrals \cite{Zyczkowski2001, Zyczkowski2003,
Andrews1999} (see Appendix~\ref{subsect:sellberg-integral} for a
quick review of Selberg's integrals). Let us define
$C^{(\alpha,\gamma)}_m=1/{\cI_m(\alpha,\gamma)}$, which are called
as normalization constants. A family of probability measures over
$\real^m_+$ can be defined as:
\begin{eqnarray}
\dif\nu_{\alpha,\gamma} (\Lambda) := C^{(\alpha,\gamma)}_m
K_\gamma(\Lambda)\prod^m_{j=1}\lambda^{\alpha-1}_j \dif\lambda_j.
\end{eqnarray}
Also, $\nu_{\alpha,\gamma}$ is a family of normalized probability
measures over the probability simplex
$$
\Delta_{m-1}:=
\Set{\Lambda=(\lambda_1,\ldots,\lambda_m)\in\real^m_+:\sum^m_{j=1}\lambda_j=1},
$$
i.e.,
\begin{eqnarray*}
\nu_{\alpha,\gamma}\Pa{\Delta_{m-1}} = \int \dif\nu_{\alpha,\gamma}(\Lambda)=1.
\end{eqnarray*}
Now a family of probability measures $\mu_{\alpha,\gamma}$ over the
set $\density{\complex^m}$ of all $m\times m$ density matrices on
$\complex^m$ can be obtained via spectral decomposition of
$\rho\in\density{\complex^m}$ with $\rho=U\Lambda U^\dagger$ as
follows
\begin{eqnarray}\label{eq:alpha-gamma}
\dif\mu_{\alpha,\gamma}(\rho) =
\dif\nu_{\alpha,\gamma}(\Lambda)\times \dif\mu_{\mathrm{Haar}}(U),
\end{eqnarray}
where
$\dif\nu_{\alpha,\gamma}(\Lambda)=\dif\nu_{\alpha,\gamma}(\lambda_1,\ldots,\lambda_m)$
and $\mu_{\mathrm{Haar}}$ is the normalized uniform Haar measure. By
definition, $\mu_{\alpha,\gamma}$ is a normalized probability
measure over $\density{\complex^m}$. In the following, we will use
this family of probability measures to calculate the average
subentropy and average coherence of randomly chosen quantum states.

%============================================================================%
\section{The average subentropy of a random mixed state}\label{sec:subentropy}
%============================================================================%

Let us consider $m$ dimensional random density matrices $\rho$
sampled according to the family of product measures
$\mu_{\alpha,\gamma}$, such that $\dif\mu_{\alpha,\gamma}(\rho) =
\dif\nu_{\alpha,\gamma}(\Lambda)\times \dif\mu_{\mathrm{Haar}}(U)$.
The subentropy of a state $\rho$ with the spectrum $\Lambda =
\set{\lambda_1,\cdots, \lambda_m}$ can be written as
\cite{JozsaB1994,Harremoes2001,Mintert2004,Datta2014} (see also
Appendix~\ref{subsect:subentropy-related})
\begin{eqnarray} \label{eq:sub1}
Q(\Lambda)&=&(-1)^{\frac{m(m-1)}2-1}\frac{\sum^m_{i=1}\lambda^m_i\ln\lambda_i\prod_{j\in\widehat
i}\phi'(\lambda_j)}{\abs{\Delta(\lambda)}^2},
\end{eqnarray}
where $\widehat i=\set{1,\ldots,m}\backslash\set{i}$,
$\phi'(\lambda_j)=\prod_{k\in\widehat j}(\lambda_j-\lambda_k)$ and
$\abs{\Delta(\lambda)}^2= \abs{\prod_{1\leqslant i<j\leqslant
m}(\lambda_i-\lambda_j)}^2$.

Subentropy of a quantum state is a quantity that arises in quantum
information theory. It is shown that subentropy provides a tight
lower bound on the accessible information for pure state ensembles
\cite{JozsaB1994}, dual to the fact that von Neumann entropy is an
upper bound in Holevo's theorem. More recently, the notion of
subentropy is generalized to its R\'{e}nyi variant
\cite{Mintert2004} and R\'{e}nyi variant of the subentropy of the
mixed state obtained by partial trace of the bipartite pure state is
interpreted as the excess of the Wehrl entropy since they are shown
to be equal. Interestingly, Mintert and \.{Z}yczkowski use these
quantities to investigate entanglement since they find that the
subentropy and the generalized R\'{e}nyi subentropies are Schur
concave, they are indeed entanglement monotones and can be used as
alternative measures of entanglement. For the convenience, we list
the properties of subentropy in the
Appendix~\ref{subsect:subentropy-related}.

The average subentropy over the set of mixed states is given by
\begin{eqnarray}
\cI^Q_m(\alpha,\gamma)=\int\dif\mu_{\alpha,\gamma}(\rho)Q(\rho)
=\int\dif\nu_{\alpha,\gamma}(\Lambda)Q(\Lambda).
\end{eqnarray}
Apparently $0\leqslant\cI^Q_m(\alpha,\gamma)\leqslant
1-\gamma_{\mathrm{Euler}}$ since the subentropy is uniformly
bounded, i.e., $0\leqslant Q(\Lambda)\leqslant
1-\gamma_{\mathrm{Euler}}$, where
$\gamma_{\mathrm{Euler}}\approx0.57722$ is Euler's constant.

\begin{prop}\label{prop1}
For $\gamma=1$ and arbitrary $\alpha$, the average
subentropy $\cI^Q_m(\alpha,1)$ is given by
\begin{eqnarray}\label{eq:gen-sub}
\cI^Q_m(\alpha,1)=\frac1{m(m+\alpha-1)}
\sum^{m-1}_{k=0} g_{mk}(\alpha) u_{mk}(\alpha),
\end{eqnarray}
where
\begin{eqnarray}
g_{mk}(\alpha)&=&\psi(m(m+\alpha-1)+1)-\psi(2(m-1)+\alpha+1-k),\label{eq:g}\\
u_{mk}(\alpha)&=&\frac{(-1)^k\Gamma(2(m-1)+\alpha+1-k)}{\Gamma(k+1)\Gamma(m-k)\Gamma(m+\alpha-1-k)},\label{eq:u}
\end{eqnarray}
with $\psi(z)=\dif \ln \Gamma(z)/\dif z$ being the digamma function.
\end{prop}

\begin{proof}
See Appendix~\ref{subsect:proposition-proof}.
\end{proof}

In the remaining, we consider the \emph{induced measure}
$\mu_{m(n)}(m\leqslant n)$ over all the $m\times m$ density matrices
of the $m$-dimensional quantum system via partial tracing over the
$n$-dimensional ancilla of uniformly Haar-distributed random
bipartite pure states of system and ancilla, which is as follows:
for $\rho=U\Lambda U^\dagger$ with
$\Lambda=\diag(\lambda_1,\ldots,\lambda_m)$ and
$U\in\rU(m)$,
\begin{eqnarray}\label{eq:m-n1}
\dif\mu_{m(n)}(\rho) =
\dif \nu_{m(n)}(\Lambda)\times \dif\mu_{\mathrm{Haar}}(U),
\end{eqnarray}
where $\dif\nu_{m(n)}(\Lambda) = C_{m(n)}
K_1(\Lambda)\prod^m_{j=1}\lambda^{n-m}_j\dif \lambda_j$
\cite{Zyczkowski2001} is the joint distribution of eigenvalues
$\Lambda=\diag(\lambda_1,\ldots,\lambda_m)$ of the density matrix
$\rho$, and $\dif\mu_{\mathrm{Haar}}(U)$ is the uniform Haar measure
over unitary group $\rU(m)$. Apparently Eq.~\eqref{eq:m-n1} is a
special case of Eq.~\eqref{eq:alpha-gamma} when
$(\alpha,\gamma)=(n-m+1,1)$. That is,
$\dif\mu_{m(n)}(\rho)=\dif\mu_{n-m+1,1}(\rho)$ and
$\dif\nu_{m(n)}(\Lambda)=\dif\nu_{n-m+1,1}(\Lambda)$. From this, we
see that
\begin{eqnarray}\label{eq:sp-coh}
\cI^Q_m(n-m+1,1)=\frac1{mn}\sum^{m-1}_{k=0}g_{mk}(n-m+1)
u_{mk}(n-m+1).
\end{eqnarray}
In fact, we find a closed-formula for the average subentropy:

\begin{lem}[Closed-form of the average subentropy]\label{lem:average-subentropy}
The average subentropy of random mixed quantum states, induced by
partial tracing the Haar-distributed random pure bipartite states in
the Hilbert space of dimension $m\ot n$ with $m\leqslant n$, is
given by the following compact formula:
\begin{eqnarray}\label{eq:ave-subentropy}
\cI^Q_m(n-m+1,1)= 1+H_{mn}-H_m-H_n.
\end{eqnarray}
\end{lem}

\begin{proof}
See Appendix~\ref{subsect:Lemma-proof}.
\end{proof}

The above expression can be rewritten as
$\cI^Q_m(n-m+1,1)=(1-\gamma_{\mathrm{Euler}})-(a_m+a_n-a_{mn})$,
where $a_k=H_k-\ln k-\gamma_{\mathrm{Euler}}$ for $k=m,n$. Since the
number series $\set{a_k}$ is monotone decreasing and approaches to
zero, it follows that $a_m+a_n-a_{mn}\geqslant0$ and
$\lim_{m\to\infty}(a_m+a_n-a_{mn})=0$ (note that $m\leqslant n$).
Based on this fact, we get that
\begin{eqnarray}
\lim_{m\to\infty}\cI^Q_m(n-m+1,1) =
1-\gamma_{\mathrm{Euler}}\approx0.42278.
\end{eqnarray}
If $m=n$, this situation corresponds to the probability measure
induced by the Hilbert-Schmidt distance \cite{Zyczkowski2001}, then
\begin{eqnarray} \label{eq:av-sub}
\cI^Q_m(1,1) =\frac1{m^2}\sum^{m-1}_{k=0}g_{mk}(1)
u_{mk}(1)=1+H_{m^2}-2H_m.
\end{eqnarray}

We find that it approaches exponentially fast towards the maximum
value of the subentropy, which is achieved for the maximally mixed
state \cite{JozsaB1994}. The maximum value of $Q(\rho)$ is
approximately equal to $0.42278$ \cite{JozsaB1994}. This is
surprising, since $Q(\rho)$ is a nonlinear function of $\rho$ and it
is not expected that the average subentropy should match with the
subentropy of the average state, which is the maximally mixed state.

%=======================================================================================%
\section{The average coherence of random mixed states and typicality}\label{sec:avg-coh}
%=======================================================================================%

Now, we are in a position to calculate the average coherence of
random mixed states and establish its typicality. Let $\rho =
U\Lambda U^\dagger$ be a mixed full-ranked quantum state on
$\complex^m$ with non-degenerate positive spectra
$\lambda_j\in\real^+(j=1,\ldots,m)$, where $\Lambda =
\diag(\lambda_1,\ldots,\lambda_m)$. Then coherence of the state
$\rho$ is given by $\sC_r(U\Lambda U^\dagger)= S(\Pi(U\Lambda
U^\dagger)) - S(\Lambda)$. The average coherence of the isospectral
density matrices can be expressed in terms of the quantum
subentropy, von Neumann entropy, and $m$-Harmonic number as follows
\cite{Cheng2015}:
\begin{eqnarray}
\overline{\sC}_r^{\mathrm{iso}}(\Lambda)&:=&\int\dif\mu_{\mathrm{Haar}}(U)\sC_r(U\Lambda
U^\dagger)= H_m - 1 + Q(\Lambda)  - S(\Lambda).
\end{eqnarray}
Here $Q(\Lambda)$ is the subentropy, given by Eq. (\ref{eq:sub1}),
$S(\Lambda)$ is the von Neumann entropy of $\Lambda$ and $H_m =
\sum_{k=1}^{m} 1/k$ is the $m$-Harmonic number. From this, we see
that the average coherence of isospectral full-ranked density
matrices depends completely on the spectrum. Also, it is known that
$0\leqslant Q(\Lambda)\leqslant 1-\gamma_{\mathrm{Euler}}$. Now,
using the product probability measures
$\dif\mu_{\alpha,\gamma}=\dif\nu_{\alpha,\gamma} \times
\mu_{\mathrm{Haar}}(U)$, the average  coherence of random mixed
states is given by
\begin{eqnarray}\label{eq:def-gen-coh}
\overline{\sC_r}(\alpha,\gamma)&:=&\int\dif\mu_{\alpha,\gamma}(\rho)\sC_r(\rho)
= \int\dif\mu_{\alpha,\gamma}(U\Lambda U^\dagger)\sC_r(U\Lambda U^\dagger)\nonumber\\
&=& H_m-1 +\cI^Q_m(\alpha,\gamma) - \cI^S_m(\alpha,\gamma),
\end{eqnarray}
where
$\cI^Q_m(\alpha,\gamma)=\int\dif\nu_{\alpha,\gamma}(\Lambda)Q(\Lambda)$
and
$\cI^S_m(\alpha,\gamma)=\int\dif\nu_{\alpha,\gamma}(\Lambda)S(\Lambda)$.
In the remaining, we again consider the \emph{induced measure}
$\mu_{m(n)}(m\leqslant n)$ over all the $m\times m$ density matrices
of the $m$-dimensional quantum system via partial tracing over the
$n$-dimensional ancilla of uniformly Haar-distributed random pure
bipartite states of system and ancilla.

\begin{thrm}[Closed-form of the average coherence]\label{thm:gen-coh}
The average coherence of random mixed states of dimension $m$
sampled from induced measures obtained via partial tracing of Haar
distributed bipartite pure states of dimension $mn$, for
$(\alpha,\gamma)=(n-m+1,1)$, is given by
\begin{eqnarray}\label{eq:generic-formula}
\overline{\sC_r}(n-m+1,1) =\frac{m-1}{2n}.
\end{eqnarray}
\end{thrm}

\begin{proof}
See Appendix~\ref{subsect:theorem-proof-coherence}.
\end{proof}

For $m=n$, which corresponds to the probability measure induced by
the Hilbert-Schmidt distance, the average coherence of random mixed
states is given by
\begin{eqnarray}\label{eq:HS1}
\overline{\sC_r}(1,1) = \frac{m-1}{2m}\to\frac12 \quad(m\to\infty).
\end{eqnarray}
The asymptotic value $\frac12$ of the average coherence is also
obtained by Pucha{\l}a \emph{et al.} \cite{Puchala2016} using free
probabilistic tools. However, free probabilistic theory can only be
used to deal with asymptotic situation. Therefore, the case where
$m\neq n$ cannot be treated by such a method. Thus their result
about average coherence is just an asymptotic value of a special
case of our formula \eqref{eq:generic-formula}. Clearly this formula
is completely exact, not asymptotically, and very simple.

From this formula, we see that for fixed $m$, the dimension of
quantum system under consideration, when $n$, the dimension of
ambient environment, is larger, the average coherence is
approximately vanishing. In fact, we have calculated the average
coherence of random pure quantum states and also established its
typicality \cite{UttamB2015}. The average coherence of random pure
quantum states in $m$-dimensional Hilbert space is given by $H_m-1$.
Thus
$$
\frac{m-1}{2m}<\frac12\ll H_m-1\to \infty.
$$
In view of this, statistically, in $m$-dimensional Hilbert space,
mixed quantum states is less useful than pure quantum states in
higher dimension when we extract quantum coherence, as a resource,
from quantum states.

Now, just like in the case of random pure states where the average coherence is a generic property of all random pure states \cite{UttamB2015}, one may ask if the
average coherence of random mixed states is also a generic property of all random mixed states. The following theorem (Theorem \ref{th:con-ent}) establishes that the
average coherence is indeed a generic property of all random mixed states, i.e., as we increase the dimension of the density matrix, almost all the density matrices generated randomly have coherence approximately equal to the average relative entropy of coherence, given by Theorem \ref{thm:gen-coh}. Thus, the average coherence of a random mixed state can be viewed as the typical coherence content of random mixed states.
\begin{thrm}\label{th:con-ent}
Let $\rho_A$ be a random mixed state on an $m$
dimensional Hilbert space $\mathcal{H}$ with $m\geqslant 3$
generated via partial tracing of the Haar distributed bipartite pure
states on $mn$ dimensional Hilbert space. Then, for all $\epsilon >
0$, the relative entropy of coherence $\sC_r(\rho_A)$ of $\rho_A$
satisfies the following inequality:
\begin{eqnarray}
\mathbf{Pr} \Set{\abs{\sC_r(\rho_A) - \frac{m-1}{2n}}> \epsilon}
\leqslant 2\exp\Pa{-\frac{m n \epsilon^2}{144 \pi^3 \ln 2 (\ln
m)^2}}.
\end{eqnarray}
\end{thrm}

\begin{proof}
See Appendix~\ref{subsect:theorem-proof-concentration}.
\end{proof}

\begin{cor}\label{cor:con-ent}
Let $\rho_A$ be a random mixed state on an $m$-dimensional Hilbert
space $\mathcal{H}$ with $m\geqslant 3$ generated via partial
tracing of the Haar distributed bipartite pure states on $m\otimes
m$ dimensional Hilbert space. Then, for all $\epsilon > 0$ and
sufficiently large $m$, when
$\abs{\overline{\sC_r}(1,1)-\frac12}<\frac\epsilon2$ (this is
equivalent to $m>\frac1{\epsilon}$), its coherence is close to the
number $\frac12$, as the deviations become exponentially rare, i.e.,
the relative entropy of coherence $\sC_r(\rho_A)$ of $\rho_A$
satisfies the following inequality:
\begin{eqnarray}
\mathbf{Pr} \Set{\abs{\sC_r(\rho_A) -\frac12}>\epsilon} \leqslant
2\exp\Pa{-\frac{m^2\epsilon^2}{576\pi^3\ln2(\ln m)^2}}.
\end{eqnarray}
\end{cor}

\begin{proof}
Since
$$
\lim_{m\to\infty}\overline{\sC_r}(1,1)=\frac12,
$$
it follows from the triangular inequality that
$$
\Set{\rho: \abs{\sC_r(\rho) -
\overline{\sC_r}(1,1)}+\abs{\overline{\sC_r}(1,1)-\frac12}<\epsilon}\subset
\Set{\rho: \abs{\sC_r(\rho) - \frac12}<\epsilon}.
$$
For dimension $m$ so large that
$\abs{\overline{\sC_r}(1,1)-\frac12}<\frac\epsilon2$, this implies
the following bounds,
\begin{eqnarray*}
\mathbf{Pr} \Set{\abs{\sC_r(\rho_A) -\frac12}<\epsilon} &\geqslant&
\mathbf{Pr}\Set{\abs{\sC_r(\rho_A) -
\overline{\sC_r}(1,1)}+\abs{\overline{\sC_r}(1,1)-\frac12}<\epsilon}\\
&\geqslant&\mathbf{Pr}\Set{\abs{\sC_r(\rho_A) -
\overline{\sC_r}(1,1)}<\epsilon-\abs{\overline{\sC_r}(1,1)-\frac12}}\\
&\geqslant& 1 -
2\exp\Pa{-\frac{m^2\Pa{\epsilon-\abs{\overline{\sC_r}(1,1)-\frac12}}^2}{144\pi^3\ln2(\ln
m)^2}},
\end{eqnarray*}
implying that
\begin{eqnarray*}
\mathbf{Pr} \Set{\abs{\sC_r(\rho_A) -\frac12}>\epsilon} \leqslant
2\exp\Pa{-\frac{m^2\epsilon^2}{576\pi^3\ln2(\ln m)^2}}.
\end{eqnarray*}
This completes the proof.
\end{proof}

Next, we present an important consequence of Theorem
\ref{th:con-ent} showing a reduction in computational complexity of
certain entanglement measures for a specific class of mixed states.

%========================================================================================%
\section{Entanglement properties of a specific class of random bipartite mixed states}
%========================================================================================%

Consider a specific class $\mathcal{X}$ of random bipartite mixed
states $\chi_{AB}$ of dimension $m\otimes m$ that are generated as
follows. First generate random mixed states for a single quantum
system $A$ in an $m$ dimensional Hilbert space via partial tracing
the Haar distributed bipartite pure states on an $mn$ dimensional
Hilbert space. Now bring in an ancilla $B$ in a fixed state
$\proj{0}_{B}$ on a $d_B$-dimensional Hilbert space and apply the
generalized CNOT gate, defined as
$$
\mathrm{CNOT}
={\sum}_{i=0}^{m-1}{\sum}_{j=0}^{m-1}\proj{i}_{A}\otimes\ket{\mathrm{mod}(i+j,m)}\bra{j}_{B}
+{\sum}_{i=0}^{m-1}{\sum}_{j=m}^{d_{B}-1}\proj{i}_{A}\otimes\proj{j}_{B},
$$
on the composite system $AB$.
% taking $A$ as the control and $B$ as the target.
The random bipartite mixed states, thus obtained, are given by
\begin{eqnarray}
% \label{eq:ran-mix}
\chi_{AB}:=\mathrm{CNOT}\Br{\rho_{A}\otimes \proj{0}_B}
=\sum_{i,j=0}^{m-1}\rho_{ij}\ket{ii}\bra{jj}_{AB},\nonumber
\end{eqnarray}
where $\rho_{A}:=
\Ptr{A_0}{\proj{\psi}_{AA_0}}=\sum_{i,j=0}^{m-1}\rho_{ij}\ket{i}\bra{j}_{A}$
is a random mixed state generated according to an induced measure
via partial tracing as mentioned above. Now, using the results on
convertibility of coherence into entanglement \cite{Alex15}, we can
estimate exactly the relative entropy of entanglement $E_r$
\cite{Vedral2002} and distillable entanglement $E_d$
\cite{Bennett1996, Rains1999} of random mixed states in the class
$\cX$. In particular,
\begin{eqnarray}
 E_r^{A|B} (\chi_{AB}) = \sC_r(\rho_{A}) =E_d^{A|B} (\chi_{AB}).
\end{eqnarray}
We can now use our exact results on the average relative entropy of
coherence of random mixed states to find the average entanglement
for the specific class of bipartite random mixed states in the class
$\cX$ as follows:
% \begin{align}
%  \overline E_r^{A|B} (\chi^{AB}) = \overline{\sC_r}(n-m+1,1) = \overline E_d^{A|B} (\chi^{AB}).
% \end{align}
\begin{eqnarray}
\overline E_r^{A|B} (\chi_{AB}) & =& \int\dif\mu_{n-m+1,1}(\rho) E_r^{A|B} \Pa{\mathrm{CNOT}\Br{\rho_A\otimes \proj{0}_B}}\nonumber\\
 & =& \int\dif\mu_{n-m+1,1}(\rho) \sC_r(\rho_A)= \overline{\sC_r}(n-m+1,1).
\end{eqnarray}
Here $\overline{\sC_r}(n-m+1,1)$ is given by Theorem
\ref{thm:gen-coh}. Similarly, $\overline E_d^{A|B}
(\chi_{AB})=\frac{m-1}{2n}$. The following corollary follows
immediately from Theorem \ref{th:con-ent}.
\begin{cor}\label{cor:coh-ent}
Let $\chi_{AB}\in\mathcal{X}$ be a random mixed
state on $m \otimes m$ dimensional Hilbert space with $m\geqslant 3$
generated as mentioned above. Then, for all $\epsilon > 0$
\begin{eqnarray}\label{Eq:conc1}
\mathbf{Pr} \Set{\abs{E_r^{A|B} (\chi_{AB}) -\frac{m-1}{2n}}>
\epsilon} \leqslant 2 \exp\Pa{-\frac{m n \epsilon^2}{144 \pi^3 \ln 2
(\ln m)^2}}
\end{eqnarray}
and
\begin{eqnarray}
\mathbf{Pr} \Set{\abs{E_d^{A|B} (\chi_{AB}) -\frac{m-1}{2n}}>
\epsilon}\leqslant 2\exp\Pa{-\frac{m n \epsilon^2}{144 \pi^3 \ln 2
(\ln m)^2}}.
\end{eqnarray}
\end{cor}

Corollary \ref{cor:coh-ent} establishes that most of the random
states in the class $\mathcal{X}$ have almost the same fixed amount
of distillable entanglement and relative entropy of entanglement in
the large $m$ limit. Thus, our results help in estimating the
entanglement content of most of the random states in the class
$\mathcal{X}$ (which is an extremely hard task), asymptotically and
show the typicality of entanglement for class $\mathcal{X}$ of mixed
states.

%============================================%
\section{Conclusion}\label{sec:conclusion}
%============================================%

To conclude, we have provided analytical expressions for the average
subentropy and the average relative entropy of coherence over the
whole set of density matrices distributed according to the family of
probability measures obtained via the spectral decomposition.
% The absence of unique probability measure on the set of mixed quantum states implies that our results depend on the probability measures on the set of mixed states that we have considered.
We also have obtained the closed-form of the average subentropy
(Lemma~\ref{lem:average-subentropy}). Based on the compact form of
the average subentropy, we find that as we increase the dimension of
the quantum system, the average subentropy approaches towards the
maximum value of subentropy (attained for the maximally mixed state)
exponentially fast, which is surprising as the subentropy is a
nonlinear function of density matrix. We also have obtained the
compact form of the average coherence
$\overline{\sC}_{m,n}=\frac{m-1}{2n}$ from combining the average
entropy formula $\overline{S}_{m,n}=H_{mn}-H_n-\frac{m-1}{2n}$
\cite{Page1993,Foong1994,Jorge1995,Sen1996} and the average
subentropy formula $\overline{Q}_{m,n}=1+H_{mn}-H_m-H_n$.
Interestingly, using L\'evy's lemma, we prove that the coherence of
random mixed states sampled from induced measures via partial
tracing show the concentration phenomenon, establishing the generic
nature of coherence content of random mixed states. As a very
important application of our results, we show a huge reduction in
the computational complexity of entanglement measures such as
relative entropy of entanglement and distillable entanglement. We
find the entanglement properties of a specific class random
bipartite mixed states, thanks to Theorem \ref{th:con-ent}. Since
quantum coherence and entanglement are deemed as useful resources
for implementations of various quantum technologies, our results
will serve as a benchmark to gauge the resourcefulness of a generic
mixed state for a certain task at hand. Furthermore, our results may
have some applications in black hole physics as to how much
coherence can be there in the Hawking radiation for non-thermal
states \cite{PageB1993}, in thermalization of closed quantum systems
and in catalytic coherence transformations.

%% Acknowledgements %%
%===========================================================================%
\noindent {\it Acknowledgments.--} L.Z. is grateful to Jie-feng Wang
for directing me to notice the combinatorial identity
\eqref{eq:wjf}, and would also like to thank the National Natural
Science Foundation of China (No.11301124 \& No.61673145) for
support. L.Z. is also supported by Natural Science Foundation of
Zhejiang Province of China (LY17A010027). U.S. acknowledges the
research fellowship of Department of Atomic Energy, Government of
India.
%===========================================================================%

%===========================================================================%

%=============================================================================%

\section{Appendix}\label{sect:appendix}

Here we give a very brief review of the subentropy and the Selberg's
integrals. Also we provide the proofs of the Proposition
\ref{prop1}, Lemma \ref{lem:average-subentropy}, and the two
Theorems \ref{thm:gen-coh} and \ref{th:con-ent} of the main text.

\subsection{Quantum Subentropy}\label{subsect:subentropy-related}

The von Neumann entropy of a quantum system is of paramount
importance in physics starting from thermodynamics
\cite{JaynesA1957,JaynesB1957} to the quantum information theory,
e.g., in studies of the classical capacity of a quantum channel and
the compressibility of a quantum source \cite{Schumacher1995}, and
serves as the least upper bound on the accessible information. The
von Neumann entropy of an $m$ dimensional density matrix $\rho$, is
defined as $S(\rho) = - \sum^m_{j=1}\lambda_j\ln\lambda_j$, where
$\lambda = \set{\lambda_1,\cdots,\lambda_m}$ are eigenvalues of
$\rho$. An analogous lower bound on the accessible information,
obtained in Ref. \cite{JozsaB1994} and called as the
\emph{subentropy} $Q(\rho)$, is defined as
$Q(\rho) =
-\sum^m_{i=1}\lambda^m_i \Pa{\prod_{j\neq i}
(\lambda_i-\lambda_j)}^{-1}\ln\lambda_i$. Also, when two or more of
the eigenvalues $\lambda_j$ are equal, the value of $Q$ is
determined by taking a limit starting with unequal eigenvalues,
unambiguously. The upper bound $S(\rho)$ and the lower bound
$Q(\rho)$ on the accessible information are achieved for the
ensemble of eigenstates of $\rho$ and the Scrooge ensemble
\cite{JozsaB1994}, respectively.
Thus,
the von Neumann entropy and the subentropy together define the range of
the accessible information for a given density matrix. Main
properties of subentropy are summarized here as the following
proposition for the reader's convenience, details can be found in
\cite{JozsaB1994, Harremoes2001, Mintert2004, Datta2014}.
\begin{prop}
The subentropy $Q(\rho)$ of a quantum state $\rho$ satisfies the
following properties:
\begin{enumerate}[(1)]
\item $0\leqslant Q(\rho)\leqslant 1-\gamma_{\mathrm{Euler}}$, where
$\gamma_{\mathrm{Euler}}\approx0.57722$ is Euler's constant. In
particular, the lower bound is achieved only at all pure states; and
the maximum value of subentropy is achieved at the maximally mixed
state, that is, $Q(\I_m/m)=\ln m - H_m + 1$, where $H_m$ is the
$m$-th harmonic number. Moreover, $Q(\rho)\leqslant
-\ln\lambda_{\max}(\rho)$, where $\lambda_{\max}(\rho)$ is the
maximum eigenvalue of $\rho$.
\item $Q(\rho)$ is a concave function in $\rho$ (which is also true for the von Neumann entropy).
That is, if $\set{\rho_j}^n_{j=1}$ are density matrices and
$\set{q_j}^n_{j=1}$ is a probability distribution, then
\begin{eqnarray}
Q\Pa{\sum^n_{j=1}q_j\rho_j}\geqslant \sum^n_{j=1}q_jQ(\rho_j).
\end{eqnarray}
In particular, $Q\Pa{\sum_i p_i U_i\rho U^\dagger_i}\geqslant
Q(\rho)$, where $\set{p_i}_i$ is a probability distribution and
$\set{U_i}_i$ is a set of unitary matrices.
\item $Q(\rho)$ is Schur-concave in $\rho$. Note that
a real-valued function $f(\rho)$ is called \emph{Schur} concave in
$\rho$  if $f(\rho)\geqslant f(\sigma)$ whenever $\rho$ is majorized
by $\sigma$.
\item  $Q(\rho)$ is a continuous function in $\rho$. Assume $\rho,\rho'\in\density{\complex^m}$. If $\norm{\rho-\rho'}_1\leqslant
e^{-1}$, then
\begin{eqnarray}
\abs{Q(\rho)-Q(\rho')} \leqslant (\ln
m)\norm{\rho-\rho'}_1+\eta(\norm{\rho-\rho'}_1),
\end{eqnarray}
where $\eta(x)=-x\ln x$.
\item $Q(\rho_A\ot\rho_B)\leqslant Q(\rho_A)+Q(\rho_B)$. In general,
$Q(\rho_A\ot\rho_B)\neq Q(\rho_A)+Q(\rho_B)$.
\end{enumerate}
\end{prop}
At present, we do not know wether the subadditivity inequality of
subentropy is true or not: $Q(\rho_{AB})\leqslant
Q(\rho_A)+Q(\rho_B)$. A similarly question can be asked for the
strong subadditivity of subentropy:
$Q(\rho_{ABC})+Q(\rho_B)\leqslant Q(\rho_{AB})+Q(\rho_{BC})$.
For a comparison between the von Neumann entropy and the subentropy, we rewrite both
$S$ and $Q$ as contour integrals. $S$ can be represented as
\begin{eqnarray}
S(\rho) = -\frac1{2\pi\mathrm{i}}\oint (\ln z)\Tr{(\I_m -
\rho/z)^{-1}}\dif z,
\end{eqnarray}
where the contour encloses all the nonzero eigenvalues of $\rho$.
$Q$ can be also represented as
\begin{eqnarray}
Q(\rho) = -\frac1{2\pi\mathrm{i}}\oint (\ln z)\det\Pa{(\I_m -
\rho/z)^{-1}}\dif z,
\end{eqnarray}
$S(\rho)$ and $Q(\rho)$ are strikingly similar in above forms and
where the trace appears in the formula for the von Neumann entropy,
the determinant appears in the formula for the subentropy. Other
comparison can also be seen in Refs. \cite{JozsaB1994,
Harremoes2001, Mintert2004, Datta2014}.
Now, we present Selberg's integrals and the calculation of the
average subentropy of random mixed states.
% \begin{widetext}

%=============================================================================%
\subsection{Selberg's Integrals and its consequences}\label{subsect:sellberg-integral}
%=============================================================================%

\begin{prop}[Selberg's Integrals,\cite{Andrews1999}]\label{prop:selberg}
If $m$ is a positive integer
and $\alpha,\beta,\gamma$ are complex numbers such that
$$
\re(\alpha)>0,~~~\re(\beta)>0,~~~\re(\gamma)>-\min\Set{\frac1m,
\frac{\re(\alpha)}{m-1}, \frac{\re(\beta)}{m-1}},
$$
then
\begin{eqnarray}
S_m(\alpha,\beta,\gamma) &=&
\int^1_0\cdots\int^1_0\Pa{\prod^m_{j=1}x^{\alpha-1}_j(1-x_j)^{\beta-1}}\abs{\Delta(x)}^{2\gamma}[\dif
x]\notag\\
&=&\prod^m_{j=1}\frac{\Gamma(\alpha+\gamma(j-1))\Gamma(\beta+\gamma(j-1))\Gamma(1+\gamma
j)}{\Gamma(\alpha+\beta+\gamma(m+j-2))\Gamma(1+\gamma)},\label{eq:selberg-int}
\end{eqnarray}
where $ \Delta(x) = \prod_{1\leqslant i<j\leqslant m}(x_i-x_j)$ and $[\dif x] = \prod^m_{j=1}\dif x_j$.
Furthermore, if $1\leqslant k\leqslant m$, then
\begin{eqnarray}
\int^1_0\cdots\int^1_0\Pa{\prod^k_{j=1}{x_j}}\Pa{\prod^m_{j=1}x^{\alpha-1}_j(1-x_j)^{\beta-1}}\abs{\Delta(x)}^{2\gamma}[\dif
x]
=S_m(\alpha,\beta,\gamma)\prod^k_{j=1}\frac{\alpha+\gamma(m-j)}{\alpha+\beta+\gamma(2m-j-1)}.\label{eq:aomoto-selberg-int}
\end{eqnarray}
\end{prop}
The following two integrals (Propositions~\ref{prop:A-1} and
\ref{prop:A-2}) are direct consequences of
Proposition~\ref{prop:selberg}.

\begin{prop}[\cite{Andrews1999}]\label{prop:A-1}
With the same conditions on the parameters $\alpha,\gamma$,
\begin{eqnarray}\label{eq:deformed-selberg}
\int^\infty_0\cdots\int^\infty_0
\abs{\Delta(x)}^{2\gamma}\prod^m_{j=1}x^{\alpha-1}_j e^{-x_j}\dif
x_j = \prod^m_{j=1}\frac{\Gamma(\alpha+\gamma(j-1))\Gamma(1+\gamma
j)}{\Gamma(1+\gamma)}.
\end{eqnarray}
\end{prop}

\begin{prop}[\cite{Andrews1999}]\label{prop:A-2}
With the same conditions on the parameters $\alpha,\gamma$, and
$1\leqslant k\leqslant m$,
\begin{eqnarray}
\int^\infty_0\cdots\int^\infty_0
\Pa{\prod^k_{j=1}x_j}\abs{\Delta(x)}^{2\gamma}\prod^m_{j=1}x^{\alpha-1}_j
e^{-x_j}\dif x_j =\Pa{\prod^k_{j=1}(\alpha+\gamma(m-j))}\Pa{
\prod^m_{j=1}\frac{\Gamma(\alpha+\gamma(j-1))\Gamma(1+\gamma
j)}{\Gamma(1+\gamma)}}.\label{eq:aomoto-deformed-selberg}
\end{eqnarray}
\end{prop}

In the following, we prove Propositions~\ref{prop:zycz} and
\ref{prop:product-delta} from Propositions~\ref{prop:A-1} and
\ref{prop:A-2}, respectively, using the Laplace transform.
\begin{prop}[\cite{Zyczkowski2001, Zyczkowski2003}]\label{prop:zycz}
It holds that
\begin{eqnarray}
\frac1{C^{(\alpha,\gamma)}_m}&:=&\int^\infty_0\cdots\int^\infty_0
\delta\Pa{1-\sum^m_{j=1}x_j}\abs{\Delta(x)}^{2\gamma}\prod^m_{j=1}x^{\alpha-1}_j\dif
x_j \notag\\
&=& \frac1{\Gamma(\alpha m+\gamma
m(m-1))}\prod^m_{j=1}\frac{\Gamma(\alpha+\gamma(j-1))\Gamma(1+\gamma
j)}{\Gamma(1+\gamma)}.\label{eq:int-1}
\end{eqnarray}
\end{prop}

\begin{proof}
Let
\begin{eqnarray*}
F(t):=\int^\infty_0\cdots\int^\infty_0 \delta\Pa{t-\sum^m_{j=1}x_j}
\abs{\Delta(x)}^{2\gamma}\prod^m_{j=1}x^{\alpha-1}_j\dif x_j.
\end{eqnarray*}
Applying the Laplace transform ($t\to s$) to $F(t)$ gives us
\begin{eqnarray*}
\widetilde F(s) &=& \int^\infty_0 F(t)e^{-st}\dif t\\
&=&\int^\infty_0\cdots\int^\infty_0 \exp\Pa{-s\sum^m_{j=1}x_j}
\abs{\Delta(x)}^{2\gamma}\prod^m_{j=1}x^{\alpha-1}_j\dif x_j\\
&=&s^{-\alpha
m-2\gamma\binom{m}{2}}\int^\infty_0\cdots\int^\infty_0\abs{\Delta(y)}^{2\gamma}\prod^m_{j=1}y^{\alpha-1}_je^{-y_j}\dif
y_j,
\end{eqnarray*}
leading to the following via the inverse Laplace transform ($s\to t$) to
$\widetilde F(s)$:
\begin{eqnarray*}
F(t) = \frac{t^{\alpha m+\gamma m(m-1)-1}}{\Gamma\Pa{\alpha m+\gamma
m(m-1)}}
\int^\infty_0\cdots\int^\infty_0\abs{\Delta(x)}^{2\gamma}\prod^m_{j=1}x^{\alpha-1}_je^{-x_j}\dif
x_j,
\end{eqnarray*}
Therefore, we have
\begin{eqnarray*}
\frac1{C^{(\alpha,\gamma)}_m}=F(1) =\frac1{\Gamma\Pa{\alpha m+\gamma
m(m-1)}}
\times\int^\infty_0\cdots\int^\infty_0\abs{\Delta(x)}^{2\gamma}\prod^m_{j=1}x^{\alpha-1}_je^{-x_j}\dif
x_j.
\end{eqnarray*}
Hence the desired identity via Eq.~\eqref{eq:deformed-selberg}.
\end{proof}

\begin{prop}\label{prop:product-delta}
It holds that, for $1\leqslant k\leqslant m$,
\begin{eqnarray}
&&\int^\infty_0\cdots\int^\infty_0
\Pa{\prod^k_{j=1}x_j}\delta\Pa{1-\sum^m_{j=1}x_j}\abs{\Delta(x)}^{2\gamma}\prod^m_{j=1}x^{\alpha-1}_j\dif x_j\nonumber\\
&&=\frac1{\Gamma(\alpha m+\gamma
m(m-1)+k)}\int^\infty_0\cdots\int^\infty_0
\Pa{\prod^k_{j=1}x_j}\abs{\Delta(x)}^{2\gamma}\prod^m_{j=1}x^{\alpha-1}_je^{-x_j}\dif
x_j.\label{eq:int-2}
\end{eqnarray}
\end{prop}

\begin{proof}
Similarly, let
\begin{eqnarray*}
f(t):=\int^\infty_0\cdots\int^\infty_0
\Pa{\prod^k_{j=1}x_j}\delta\Pa{t-\sum^m_{j=1}x_j}\abs{\Delta(x)}^{2\gamma}\prod^m_{j=1}x^{\alpha-1}_j\dif
x_j.
\end{eqnarray*}
Then, the Laplace transform of $f(t)$ is given by
\begin{eqnarray*}
\widetilde f(s)&&=\int^\infty_0\cdots\int^\infty_0
\Pa{\prod^k_{j=1}x_j}\exp\Pa{-\sum^m_{j=1}sx_j}\abs{\Delta(x)}^{2\gamma}\prod^m_{j=1}x^{\alpha-1}_j\dif
x_j\\
&&=s^{-(\alpha m+\gamma
m(m-1)+k)}\int^\infty_0\cdots\int^\infty_0
\Pa{\prod^k_{j=1}y_j}\abs{\Delta(y)}^{2\gamma}\prod^m_{j=1}y^{\alpha-1}_je^{-y_j}\dif
y_j.
\end{eqnarray*}
% Hence
% \begin{eqnarray*}
% \widetilde f(s):=s^{-(\alpha m+\gamma
% m(m-1)+k)}\int^\infty_0\cdots\int^\infty_0
% \Pa{\prod^k_{j=1}y_j}\abs{\Delta(y)}^{2\gamma}\prod^m_{j=1}y^{\alpha-1}_je^{-y_j}\dif
% y_j.
% \end{eqnarray*}
Therefore, we have
\begin{eqnarray*}
f(t):=\frac{t^{\alpha m+\gamma m(m-1)+k-1}}{\Gamma(\alpha m+\gamma
m(m-1)+k)}\int^\infty_0\cdots\int^\infty_0
\Pa{\prod^k_{j=1}y_j}\abs{\Delta(y)}^{2\gamma}\prod^m_{j=1}y^{\alpha-1}_je^{-y_j}\dif
y_j.
\end{eqnarray*}
By setting $t=1$ in the above equation, we derived the desired
identity via Eq.~\eqref{eq:aomoto-deformed-selberg}.
\end{proof}
\begin{prop}
It holds that
\begin{eqnarray}
\frac{\dif}{\dif t}\Pa{\frac{\Gamma(t+a)}{\Gamma(t+b)}} =
\Pa{\psi(t+a)-\psi(t+b)}\frac{\Gamma(t+a)}{\Gamma(t+b)},
\end{eqnarray}
where $\psi(t)=\frac{\dif}{\dif t}\ln \Gamma(t)$.
\end{prop}

\subsection{The proof of Proposition~\ref{prop1} of the main text}\label{subsect:proposition-proof}

A family of probability measures over $\real^m_+$ can be defined as:
\begin{align}
\dif\nu_{\alpha,\gamma} (\Lambda):=C^{(\alpha,\gamma)}_m K_\gamma(\Lambda)\prod^m_{j=1}\lambda^{\alpha-1}_j \dif\lambda_j,
\end{align}
where $K_1(\Lambda)$ is given by
\begin{align}
%\label{eq:k-lambda}
 K_1(\Lambda)=\delta\Pa{1-\sum^m_{j=1}\lambda_j}\abs{\Delta(\lambda)}^{2},
\end{align}
with
$\Delta(\lambda)=\prod_{1\leqslant i<j\leqslant
m}(\lambda_i-\lambda_j)$ and $C^{(\alpha,\gamma)}_m=1/{\cI_m(\alpha,\gamma)}$ with

\begin{align}
%\label{eq:alpha-beta}
\cI_m(\alpha,\gamma)=\frac{1}{\Gamma\Pa{\alpha m+\gamma m(m-1)}}\prod^m_{j=1}\frac{\Gamma\Pa{\alpha+\gamma(j-1)}\Gamma\Pa{1+\gamma
j}}{\Gamma\Pa{1+\gamma}}.
\end{align}
The subentropy of a state $\rho$ with the spectrum $\Lambda =\set{\lambda_1,\cdots, \lambda_m}$ can be written as \cite{JozsaB1994, Harremoes2001, Mintert2004, Datta2014}
\begin{align}
\label{eq:sub}
Q(\Lambda)
&=(-1)^{\frac{m(m-1)}2-1}\frac{\sum^m_{i=1}\lambda^m_i\ln\lambda_i\prod_{j\in\widehat
i}\phi'(\lambda_j)}{\abs{\Delta(\lambda)}^2},
% &=(-1)^{\frac{m(m-1)}2-1}\frac1{\abs{\Delta(\lambda)}^2}\left.\frac{\dif}{\dif
% t}\Pa{\sum^m_{i=1}\lambda^t_i\prod_{j\in\widehat i}\phi'(\lambda_j)}\right|_{t=m},
\end{align}
where $\widehat i=\set{1,\ldots,m}\backslash\set{i}$, $\phi'(\lambda_j)=\prod_{k\in\widehat
j}(\lambda_j-\lambda_k)$ and
$\abs{\Delta(\lambda)}^2= \abs{\prod_{1\leqslant i<j\leqslant m}(\lambda_i-\lambda_j)}^2$.
The average subentropy over the set of mixed state is given by
\begin{align}
&
\cI^Q_m(\alpha,\gamma)=\int\dif\mu_{\alpha,\gamma}(\rho)Q(\rho)=\int\dif\nu_{\alpha,\gamma}(\Lambda)Q(\Lambda).
\end{align}
Denote $\phi(x):=\prod^m_{j=1}(x-x_j)$. Then $\phi'(x) =
\sum^m_{i=1}\prod_{j\in\widehat i}(x-x_j)$. Thus
$\phi'(x_i)=\prod_{j\in\widehat i}(x_i-x_j)$. Furthermore, we have
\begin{eqnarray}
\prod^m_{i=1}\phi'(x_i) = \prod^m_{i=1}\prod_{j\in\widehat
i}(x_i-x_j) = (-1)^{\frac{m(m-1)}2}\abs{\Delta(x)}^2.
\end{eqnarray}
Here $\abs{\Delta(x)}^2=\abs{\Delta(x_1,\ldots,x_m)}^2$ is called
the \emph{discriminant} of $\phi$ \cite{Andrews1999}. We also have
\begin{eqnarray}
\phi'(\lambda_2)\cdots \phi'(\lambda_m) =
(-1)^{\frac{m(m-1)}2}\phi'(\lambda_1)\abs{\Delta(\lambda_2,\ldots,\lambda_m)}^2.
\end{eqnarray}
If we expand the polynomial $\phi(x)$, then we have:
\begin{eqnarray}
\phi(x) =
x^m-\Pa{\sum^m_{j=1}x_j}x^{m-1}+\cdots+(-1)^m\prod^m_{j=1}x_j=\sum^m_{j=0}
(-1)^j e_jx^{m-j},
\end{eqnarray}
where $e_j(j=1,\ldots,m)$ is the $j$-th elementary symmetric
polynomial in $x_1,\ldots,x_m$, with $e_0\equiv1$.

In what follows, we calculate the integral $\cI^Q_m(\alpha,\gamma)$
for $\gamma=1$. Propositions~\ref{prop:zycz} and
\ref{prop:product-delta} will be used frequently for
$\gamma=1$.
\begin{eqnarray*}
\cI^Q_m(\alpha,1)&&=-mC^{(\alpha,1)}_m\sum^{m-1}_{k=0}(-1)^k\int^1_0\dif\lambda_1\lambda^{2(m-1)+\alpha-k}_1\ln\lambda_1\nonumber\\
&&~~~~~\times\int^\infty_0\cdots \int^\infty_0e_k
\delta\Pa{(1-\lambda_1)-\sum^m_{j=2}\lambda_j}\abs{\Delta(\lambda_2,\ldots,\lambda_m)}^2\prod^m_{j=2}\lambda^{\alpha-1}_j\dif\lambda_j.
\end{eqnarray*}
It suffices to calculate a family of integrals in terms
of the following form: for $k=0,1,\ldots,m-1$,
\begin{eqnarray*}
\int^\infty_0\cdots\int^\infty_0 e_k
\delta\Pa{(1-\lambda_1)-\sum^m_{j=2}\lambda_j}\abs{\Delta(\lambda_2,\ldots,\lambda_m)}^2\prod^m_{j=2}\lambda^{\alpha-1}_j\dif\lambda_j.
\end{eqnarray*}
If $k=0$, then
\begin{eqnarray}
&&\int^\infty_0\cdots\int^\infty_0e_0
\delta\Pa{(1-\lambda_1)-\sum^m_{j=2}\lambda_j}\abs{\Delta(\lambda_2,\ldots,\lambda_m)}^2\prod^m_{j=2}\lambda^{\alpha-1}_j\dif\lambda_j\notag\\
&&=(1-\lambda_1)^{(m-1)(m+\alpha-2)-1}\int^\infty_0
\delta\Pa{1-\sum^{m-1}_{j=1}x_j}
\abs{\Delta(x_1,\ldots,x_{m-1})}^2\prod^{m-1}_{j=1}x^{\alpha-1}_j\dif
x_j\notag\\
&&=(1-\lambda_1)^{(m-1)(m+\alpha-2)-1}
\frac{\prod^{m-1}_{j=1}\Gamma(\alpha+j-1)\Gamma(1+j)}{\Gamma((m-1)(m+\alpha-2))}.
\end{eqnarray}
Here we used Proposition~\ref{prop:zycz} in the last equality.

If $1\leqslant k\leqslant m-1$, it suffices to calculate the
following:
\begin{eqnarray}
&&\int^\infty_0\cdots\int^\infty_0 \Pa{\prod^k_{j=1}\lambda_{j+1}}
\delta\Pa{(1-\lambda_1)-\sum^m_{j=2}\lambda_j}\abs{\Delta(\lambda_2,\ldots,\lambda_m)}^2\prod^m_{j=2}\lambda^{\alpha-1}_j\dif\lambda_j\notag\\
&&=(1-\lambda_1)^{(m-1)(m+\alpha-2)+k-1}\times\notag\\
&&~~~\int^\infty_0\cdots\int^\infty_0\Pa{\prod^k_{j=1}x_j}\delta\Pa{1-\sum^{m-1}_{j=1}x_j}
\abs{\Delta(x_1,\ldots,x_{m-1})}^2\prod^{m-1}_{j=1}x^{\alpha-1}_j\dif
x_j\notag\\
&&=(1-\lambda_1)^{(m-1)(m+\alpha-2)+k-1}\frac{\prod^{m-1}_{j=1}\Gamma(\alpha+j-1)\Gamma(1+j)}{\Gamma((m-1)(m+\alpha-2)+k)}\frac{\Gamma(m+\alpha-1)}{\Gamma(m+\alpha-1-k)}.
\end{eqnarray}
Here we used Proposition~\ref{prop:product-delta}.
Next, we calculate the integral
\begin{eqnarray*}
\int^1_0 \dif\lambda_1\lambda^t_1
\int^\infty_0\cdots\int^\infty_0e_k
\delta\Pa{(1-\lambda_1)-\sum^m_{j=2}\lambda_j}\abs{\Delta(\lambda_2,\ldots,\lambda_m)}^2\prod^m_{j=2}\lambda^{\alpha-1}_j\dif\lambda_j.
\end{eqnarray*}
(1). If $k=0$, then
\begin{eqnarray*}
&&\int^1_0 \dif\lambda_1\lambda^t_1\int^\infty_0\cdots\int^\infty_0
\delta\Pa{(1-\lambda_1)-\sum^m_{j=2}\lambda_j}\abs{\Delta(\lambda_2,\ldots,\lambda_m)}^2\prod^m_{j=2}\lambda^{\alpha-1}_j\dif\lambda_j\\
&&=\frac{\prod^{m-1}_{j=1}\Gamma(\alpha+j-1)\Gamma(1+j)}{\Gamma((m-1)(m+\alpha-2))}\times
\int^1_0
\lambda^t_1(1-\lambda_1)^{(m-1)(m+\alpha-2)-1}\dif\lambda_1\\
&&=\frac{\prod^{m-1}_{j=1}\Gamma(\alpha+j-1)\Gamma(1+j)}{\Gamma((m-1)(m+\alpha-2))}\times\frac{\Gamma(t+1)\Gamma((m-1)(m+\alpha-2))}{\Gamma(t+1+(m-1)(m+\alpha-2))}\\
&&=\frac{\Gamma(t+1)\prod^{m-1}_{j=1}\Gamma(\alpha+j-1)\Gamma(1+j)}{\Gamma(t+1+(m-1)(m+\alpha-2))}.
\end{eqnarray*}
By taking the derivative with respect to $t$ on both sides, we get
\begin{eqnarray*}
&&\int^1_0
\dif\lambda_1\lambda^t_1\ln\lambda_1\int^\infty_0\cdots\int^\infty_0
\delta\Pa{(1-\lambda_1)-\sum^m_{j=2}\lambda_j}\abs{\Delta(\lambda_2,\ldots,\lambda_m)}^2\prod^m_{j=2}\lambda^{\alpha-1}_j\dif\lambda_j\\
&&=\Br{\psi(t+1)-\psi(t+1+(m-1)(m+\alpha-2))}\frac{\Gamma(t+1)\prod^{m-1}_{j=1}\Gamma(\alpha+j-1)\Gamma(1+j)}{\Gamma(t+1+(m-1)(m+\alpha-2))}.
\end{eqnarray*}
For $t=2(m-1)+\alpha$, we have
\begin{eqnarray*}
&&\int^1_0
\dif\lambda_1\lambda^{2(m-1)+\alpha}_1\ln\lambda_1\int^\infty_0\cdots\int^\infty_0
\delta\Pa{(1-\lambda_1)-\sum^m_{j=2}\lambda_j}\abs{\Delta(\lambda_2,\ldots,\lambda_m)}^2\prod^m_{j=2}\lambda^{\alpha-1}_j\dif\lambda_j\\
&&=\Br{\psi(2(m-1)+\alpha+1)-\psi(m(m+\alpha-1)+1)}\frac{\Gamma(2(m-1)+\alpha+1)\prod^{m-1}_{j=1}\Gamma(\alpha+j-1)\Gamma(1+j)}{\Gamma(m(m+\alpha-1)+1)}.
\end{eqnarray*}

(2). If $1\leqslant k\leqslant m-1$, then
\begin{eqnarray*}
&&\int^1_0 \dif\lambda_1\lambda^t_1
\int^\infty_0\cdots\int^\infty_0e_k
\delta\Pa{(1-\lambda_1)-\sum^m_{j=2}\lambda_j}\abs{\Delta(\lambda_2,\ldots,\lambda_m)}^2\prod^m_{j=2}\lambda^{\alpha-1}_j\dif\lambda_j\\
&&=\binom{m-1}{k}\int^1_0
\dif\lambda_1\lambda^t_1\int^\infty_0\cdots\int^\infty_0\Pa{\prod^k_{j=1}\lambda_{j+1}}\delta\Pa{(1-\lambda_1)-\sum^m_{j=2}\lambda_j}\abs{\Delta(\lambda_2,\ldots,\lambda_m)}^2\prod^m_{j=2}\lambda^{\alpha-1}_j\dif\lambda_j\\
&&=\binom{m-1}{k}\frac{\prod^{m-1}_{j=1}\Gamma(\alpha+j-1)\Gamma(1+j)}{\Gamma((m-1)(m+\alpha-2)+k)}\prod^k_{j=1}(m+\alpha-j-1)\times\int^1_0
\lambda^t_1 (1-\lambda_1)^{(m-1)(m+\alpha-2)+k-1}\dif\lambda_1\\
&&=\binom{m-1}{k}\frac{\prod^{m-1}_{j=1}\Gamma(\alpha+j-1)\Gamma(1+j)}{\Gamma((m-1)(m+\alpha-2)+k)}\prod^k_{j=1}(m+\alpha-j-1)\times\frac{\Gamma(t+1)\Gamma((m-1)(m+\alpha-2)+k)}{\Gamma(t+1+(m-1)(m+\alpha-2)+k)}\\
&&=\binom{m-1}{k}\frac{\Gamma(t+1)\prod^{m-1}_{j=1}\Gamma(\alpha+j-1)\Gamma(1+j)}{\Gamma(t+1+(m-1)(m+\alpha-2)+k)}\prod^k_{j=1}(m+\alpha-j-1)\\
&&=\binom{m-1}{k}\frac{\Gamma(t+1)\prod^{m-1}_{j=1}\Gamma(\alpha+j-1)\Gamma(1+j)}{\Gamma(t+1+(m-1)(m+\alpha-2)+k)}\frac{\Gamma(m+\alpha-1)}{\Gamma(m+\alpha-1-k)}.
\end{eqnarray*}
By taking the derivative with respect to $t$, we get
\begin{eqnarray*}
&&\int^1_0 \dif\lambda_1\lambda^t_1\ln\lambda_1
\int^\infty_0\cdots\int^\infty_0e_k
\delta\Pa{(1-\lambda_1)-\sum^m_{j=2}\lambda_j}\abs{\Delta(\lambda_2,\ldots,\lambda_m)}^2\prod^m_{j=2}\lambda^{\alpha-1}_j\dif\lambda_j\\
&&=\binom{m-1}{k}\Br{\psi(t+1)-\psi(t+1+(m-1)(m+\alpha-2)+k)}\\
&&~~~\times\frac{\Gamma(t+1)\prod^{m-1}_{j=1}\Gamma(\alpha+j-1)\Gamma(1+j)}{\Gamma(t+1+(m-1)(m+\alpha-2)+k)}\frac{\Gamma(m+\alpha-1)}{\Gamma(m+\alpha-1-k)}.
\end{eqnarray*}
For $t=2(m-1)+\alpha-k$, we have
\begin{eqnarray*}
&&\int^1_0 \dif\lambda_1\lambda^{2(m-1)+\alpha-k}_1\ln\lambda_1
\int^\infty_0\cdots\int^\infty_0e_k
\delta\Pa{(1-\lambda_1)-\sum^m_{j=2}\lambda_j}\abs{\Delta(\lambda_2,\ldots,\lambda_m)}^2\prod^m_{j=2}\lambda^{\alpha-1}_j\dif\lambda_j\\
&&=\binom{m-1}{k}\Br{\psi(2(m-1)+\alpha-k+1)-\psi(m(m+\alpha-1)+1)}\\
&&~~~\times\frac{\Gamma(2(m-1)+\alpha-k+1)\prod^{m-1}_{j=1}\Gamma(\alpha+j-1)\Gamma(1+j)}{\Gamma(m(m+\alpha-1)+1)}\frac{\Gamma(m+\alpha-1)}{\Gamma(m+\alpha-1-k)}.\\
\end{eqnarray*}

In summary, we get
\begin{eqnarray}
\label{eq:bef-fin}
&&\cI^Q_m(\alpha,1)\nonumber\\
&&=-mC^{(\alpha,1)}_m\left[\int^1_0\dif\lambda_1\lambda^{2(m-1)+\alpha}_1\ln\lambda_1e_0
\int^\infty_0\cdots\int^\infty_0\delta\Pa{(1-\lambda_1)-\sum^m_{j=2}\lambda_j}\abs{\Delta(\lambda_2,\ldots,\lambda_m)}^2\prod^m_{j=2}\lambda^{\alpha-1}_j\dif\lambda_j\right.\nonumber\\
&&~~~\left.+\sum^{m-1}_{k=1}(-1)^k\int^1_0\dif\lambda_1\lambda^{2(m-1)+\alpha-k}_1\ln\lambda_1
\int^\infty_0 \cdots\int^\infty_0e_k
\delta\Pa{(1-\lambda_1)-\sum^m_{j=2}\lambda_j}\abs{\Delta(\lambda_2,\ldots,\lambda_m)}^2\prod^m_{j=2}\lambda^{\alpha-1}_j\dif\lambda_j\right]\nonumber\\
&&=-mC^{(\alpha,1)}_m\left[\sum^{m-1}_{k=0}(-1)^k\binom{m-1}{k}\Br{\psi(2(m-1)+\alpha-k+1)-\psi(m(m+\alpha-1)+1)}\right.\nonumber\\
&&~~~~~~~~~~~~~~~~~~\left.\times\frac{\Gamma(2(m-1)+\alpha-k+1)\prod^{m-1}_{j=1}\Gamma(\alpha+j-1)\Gamma(1+j)}{\Gamma(m(m+\alpha-1)+1)}\frac{\Gamma(m+\alpha-1)}{\Gamma(m+\alpha-1-k)}\right]\nonumber\\
&&=-\frac{mC^{(\alpha,1)}_m{\Gamma(m+\alpha-1)}\prod^{m-1}_{j=1}\Gamma(\alpha+j-1)\Gamma(1+j)}{\Gamma(m(m+\alpha-1)+1)}\nonumber\\
&&~~~~~~~~~\times\sum^{m-1}_{k=0}(-1)^k\binom{m-1}{k}\Br{\psi(2(m-1)+\alpha+1-k)-\psi(m(m+\alpha-1)+1)}\frac{\Gamma(2(m-1)+\alpha+1-k)}{\Gamma(m+\alpha-1-k)}\nonumber\\
&&=-\frac1{m(m+\alpha-1)}\left[\sum^{m-1}_{k=0}(-1)^k\Br{\psi(2(m-1)+\alpha+1-k)-\psi(m(m+\alpha-1)+1)}\right.\notag\\
&&~~~~~~~~~~~~~~~~~~~~~~~~~~~~~~~~~~\times\left.\frac{\Gamma(2(m-1)+\alpha+1-k)}{\Gamma(k+1)\Gamma(m-k)\Gamma(m+\alpha-1-k)}\right].
\end{eqnarray}
Let us define
\begin{align}
\label{eq:g1} g_{mk}(\alpha)
=&\psi(m(m+\alpha-1)+1)-\psi(2(m-1)+\alpha+1-k),
\end{align}
and
\begin{align}
\label{eq:u1}
 u_{mk}(\alpha)=\frac{(-1)^k\Gamma(2(m-1)+\alpha+1-k)}{\Gamma(k+1)\Gamma(m-k)\Gamma(m+\alpha-1-k)}.
\end{align}
Then, from Eq. (\ref{eq:bef-fin}), we have
 \begin{align}
%\label{eq:gen-sub}
\cI^Q_m(\alpha,1)=\frac1{m(m+\alpha-1)} \sum^{m-1}_{k=0}
g_{mk}(\alpha) u_{mk}(\alpha).
\end{align}
This completes the proof of Proposition \ref{prop1} of main text. For $(\alpha,\gamma)=(n-m+1,1)$, we have
\begin{align}
\label{eq:sp-coh1} &\cI^Q_m(n-m+1,1)
=\frac1{mn}\sum^{m-1}_{k=0}g_{mk}(n-m+1) u_{mk}(n-m+1).
\end{align}
If $m=n$, this situation corresponds to the measure induced by the
Hilbert-Schmidt distance \cite{Zyczkowski2001}, then we have
\begin{align}
\label{eq:av-sub1} \cI^Q_m(1,1)
&=\frac1{m^2}\sum^{m-1}_{k=0}g_{mk}(1) u_{mk}(1).
\end{align}
In Eqs. (\ref{eq:sp-coh1}) and (\ref{eq:av-sub1}), the functions $g_{mk}$ and $u_{mk}$
are given by Eqs. (\ref{eq:g1}) and (\ref{eq:u1}).

\subsection{The proof of Lemma \ref{lem:average-subentropy} of the main
text}\label{subsect:Lemma-proof}

Now that
\begin{eqnarray}
\cI^Q_m(n-m+1,1) =
\frac1{mn}\sum^{m-1}_{k=0}\frac{(-1)^k\Gamma(m+n-k)}{k!\Gamma(m-k)\Gamma(n-k)}\Br{\psi(mn+1)-\psi(m+n-k)},
\end{eqnarray}
by the fact that $\psi(\nu+1)=H_\nu-\gamma_{\text{Euler}}$ for any
positive integer $\nu$, it follows that
\begin{eqnarray}
\cI^Q_m(n-m+1,1) =
\frac1{mn}\sum^{m-1}_{k=0}\frac{(-1)^k\Gamma(m+n-k)}{k!\Gamma(m-k)\Gamma(n-k)}(H_{mn}-H_{m+n-1-k}).
\end{eqnarray}
To obtain the closed formula: $\cI^Q(n-m+1,1)=1+H_{mn}-H_m-H_n$, it
suffices to prove the following identities:
\begin{eqnarray}
\sum^{m-1}_{k=0}\frac{(-1)^k\Gamma(m+n-k)}{k!\Gamma(m-k)\Gamma(n-k)}&=&mn,\\
\sum^{m-1}_{k=0}\frac{(-1)^k\Gamma(m+n-k)}{k!\Gamma(m-k)\Gamma(n-k)}H_{m+n-1-k}&=&mn(H_m+H_n-1).
\end{eqnarray}
To this end, we need to fix some notations.
$$
\binom{z}{n} := \frac{z(z-1)\cdots
(z-n+1)}{n!}~~\text{for}~~n\in\natural~~\text{and arbitrary}~~z.
$$
Binomial coefficients can be generalized to multinomial coefficients
which is defined to be the number:
$$
\binom{n}{m_1,m_2,\ldots,m_q} = \frac{n!}{m_1!m_2!\cdots m_q!},
$$
where $n=\sum^q_{j=1}m_j$. For any real $z\in\real$ and positive
integer pairs $(m,n)$ with $m\leqslant n$, it holds
that\footnote{This identity can be referred to the link:
\url{https://en.wikipedia.org/wiki/Binomial_coefficient}}
\begin{eqnarray}\label{eq:wjf}
\binom{z}{m}\binom{z}{n} = \sum^m_{k=0}
\binom{m+n-k}{k,m-k,n-k}\binom{z}{m+n-k}.
\end{eqnarray}
This identity had appeared earlier in Gould's book \cite[Eq.~(6.44),
pp57]{Gould1972}, and it was due to Riordan \cite{Riordan1979}. We
know that, for any given real $z$ and positive integer $k$,
\begin{eqnarray}
\binom{-1}{k}&=&(-1)^k,\\
\binom{z}{k}&=&\frac zk\binom{z-1}{k-1},\\
\frac{\dif}{\dif z}\binom{z}{k} &=&
\binom{z}{k}\sum^{k-1}_{i=0}\frac1{z-i}.
\end{eqnarray}
Then
\begin{eqnarray}
\frac{\dif}{\dif z}\binom{z-1}{m-1} &=&
\binom{z-1}{m-1}\sum^{m-2}_{i=0}\frac1{(z-1)-i} =
\binom{z-1}{m-1}\sum^{m-1}_{i=1}\frac1{z-i},
\end{eqnarray}
implying that
\begin{eqnarray}
\left.\frac{\dif}{\dif z}\right|_{z=-1}\binom{z-1}{m-1} =
\binom{-2}{m-1}\sum^{m-1}_{i=1}\frac1{-1-i} = (-1)^m m(H_m - 1).
\end{eqnarray}
Similarly
\begin{eqnarray}
\left.\frac{\dif}{\dif z}\right|_{z=-1}\binom{z-1}{n-1} =
\binom{-2}{n-1}\sum^{n-1}_{i=1}\frac1{-1-i} = (-1)^n n(H_n - 1).
\end{eqnarray}
We also see that
\begin{eqnarray}
\left.\frac{\dif}{\dif z}\right|_{z=-1}\binom{z}{k} =
\binom{-1}{k}\sum^{k-1}_{i=0}\frac1{-1-i} = (-1)^{k+1}H_k.
\end{eqnarray}
Note that
\begin{eqnarray*}
\binom{z}{m}\binom{z}{n} = \sum^m_{k=0}
\frac{(m+n-k)!}{(m-k)!(n-k)!k!}\binom{z}{m+n-k},
\end{eqnarray*}
Replacement of $(m,n)$ by $(m-1,n-1)$ gives rise to
\begin{eqnarray*}
\binom{z}{m-1}\binom{z}{n-1} &=& \sum^{m-1}_{k=0}
\frac{(m+n-2-k)!}{(m-1-k)!(n-1-k)!k!}\binom{z}{m+n-2-k}\\
&=& \sum^{m-1}_{k=0}
\frac{(m+n-1-k)!}{(m-1-k)!(n-1-k)!k!}\cdot\frac1{n+m-1-k}\binom{z}{m+n-2-k}\\
&=& \sum^{m-1}_{k=0}
\frac{\Gamma(m+n-k)}{\Gamma(m-k)\Gamma(n-k)k!}\cdot\frac1{n+m-1-k}\binom{z}{m+n-2-k},
\end{eqnarray*}
that is,
\begin{eqnarray*}
\binom{z}{m-1}\binom{z}{n-1} = \sum^{m-1}_{k=0}
\frac{\Gamma(m+n-k)}{\Gamma(m-k)\Gamma(n-k)k!}\cdot\frac1{n+m-1-k}\binom{z}{m+n-2-k},
\end{eqnarray*}
multiplying both sides by $(z+1)$, we get
\begin{eqnarray*}
(z+1)\binom{z}{m-1}\binom{z}{n-1} = \sum^{m-1}_{k=0}
\frac{\Gamma(m+n-k)}{\Gamma(m-k)\Gamma(n-k)k!}\cdot\frac{z+1}{n+m-1-k}\binom{(z+1)-1}{(m+n-1-k)-1}.
\end{eqnarray*}
Therefore
\begin{eqnarray*}
(z+1)\binom{z}{m-1}\binom{z}{n-1} = \sum^{m-1}_{k=0}
\frac{\Gamma(m+n-k)}{\Gamma(m-k)\Gamma(n-k)k!}\cdot\binom{z+1}{m+n-1-k}.
\end{eqnarray*}
Again, replacement of $z+1$ by $z$ gives rise to
\begin{eqnarray}\label{eq:AAAAA}
z\binom{z-1}{m-1}\binom{z-1}{n-1} = \sum^{m-1}_{k=0}
\frac{\Gamma(m+n-k)}{\Gamma(m-k)\Gamma(n-k)k!}\cdot\binom{z}{m+n-1-k}.
\end{eqnarray}
(i). Letting $z=-1$ in Eq.~\eqref{eq:AAAAA} gives rise to
\begin{eqnarray}
\sum^{m-1}_{k=0}
\frac{\Gamma(m+n-k)}{\Gamma(m-k)\Gamma(n-k)k!}(-1)^{m+n-1-k}
=-\binom{-2}{m-1}\binom{-2}{n-1} = (-1)^{m+n-1}mn.
\end{eqnarray}
Hence,
\begin{eqnarray}
\sum^{m-1}_{k=0}
\frac{(-1)^k\Gamma(m+n-k)}{k!\Gamma(m-k)\Gamma(n-k)} =mn.
\end{eqnarray}
(ii). Furthermore, differentiating both sides at $z=-1$ leads to
\begin{eqnarray*}
&&\sum^{m-1}_{k=0}
\frac{\Gamma(m+n-k)}{\Gamma(m-k)\Gamma(n-k)k!}(-1)^{m+n-k}H_{m+n-1-k}\\
&&= \left.\Br{\binom{z-1}{m-1}\binom{z-1}{n-1}\Pa{1+z\sum^{m-1}_{i=1}\frac1{z-i}+z\sum^{n-1}_{i=1}\frac1{z-i}}}\right|_{z=-1}\\
&&=(-1)^{m+n}mn(H_m+H_n-1).
\end{eqnarray*}
Finally, we divide by $(-1)^{m+n}$ on both sides and get the
conclusion:
\begin{eqnarray}
&&\sum^{m-1}_{k=0}
\frac{\Gamma(m+n-k)}{\Gamma(m-k)\Gamma(n-k)}\cdot\frac{(-1)^k}{k!}\cdot
H_{m+n-1-k} = mn(H_m+H_n-1).
\end{eqnarray}
Hence the result.

\subsection{The proof of Theorem~\ref{thm:gen-coh} of the main
text}\label{subsect:theorem-proof-coherence}

For $(\alpha,\gamma)=(n-m+1,1)$
the value of average subentropy $\cI^Q_m(n-m+1,1)$ is given by Eq.~\eqref{eq:sp-coh1}. From the results of Page \cite{Page1993} and others \cite{Foong1994, Jorge1995, Sen1996} it is also known that
\begin{eqnarray}
\label{eq:av-ent} \cI^S_m(n-m+1,1)=H_{mn}-H_n-\frac{m-1}{2n}.
\end{eqnarray}
Let $a_n=H_n-\ln n-\gamma_{\mathrm{Euler}}$. Clearly
$\lim_{n\to\infty}a_n=0$. Now
$$
\cI^S_m(n-m+1,1)= \Pa{\ln n - \frac{m-1}{2n}} + (a_{mn}-a_n).
$$
We get $\cI^S_m(1,1)\asymp \ln m-\frac12$ when $m$ becomes very
large.

The average  coherence of random mixed states is given by
\begin{eqnarray}
\label{eq:def-gen-coh1}
\overline{\sC_r}(\alpha,\gamma)&:=&\int\dif\mu_{\alpha,\gamma}(\rho)\sC_r(\rho)
= \int\dif\mu_{\alpha,\gamma}(U\Lambda U^\dagger)\sC_r(U\Lambda U^\dagger)\nonumber\\
&=&\int\dif\nu_{\alpha,\gamma}(\Lambda)\Br{\int\dif\mu_{\mathrm{Haar}}(U)S(\Pi(U\Lambda
U^\dagger)) - S(\Lambda)}\nonumber\\
&=& H_m-1 + \int\dif\nu_{\alpha,\gamma}(\Lambda)(Q(\Lambda) - S(\Lambda))\nonumber\\
&=& H_m-1 +\cI^Q_m(\alpha,\gamma) - \cI^S_m(\alpha,\gamma),
\end{eqnarray}
where $\cI^Q_m(\alpha,\gamma)=\int\dif\nu_{\alpha,\gamma}(\Lambda)Q(\Lambda)$ and $\cI^S_m(\alpha,\gamma)=\int\dif\nu_{\alpha,\gamma}(\Lambda)S(\Lambda)$. Also, we have used the fact that the average coherence of the isosepectral density matrices can be expressed in terms of the quantum subentropy, von Neumann entropy, and $m$-Harmonic number as follows \cite{Cheng2015}:
\begin{eqnarray*}
\overline{\sC}_r^{\mathrm{iso}}(\Lambda)&&:=\int\dif\mu_{\mathrm{Haar}}(U)\sC_r(U\Lambda U^\dagger)\nonumber\\
&&= H_m - 1 + Q(\Lambda)  - S(\Lambda).
\end{eqnarray*}
Here $Q(\Lambda)$ is the subentropy, given by Eq. (\ref{eq:sub}),
$S(\Lambda)$ is the von Neumann entropy of $\Lambda$ and $H_m =
\sum_{k=1}^{m} 1/k$ is the $m$-Harmonic number. Now using
Eqs.~\eqref{eq:sp-coh1} and (\ref{eq:av-ent}), in Eq.
(\ref{eq:def-gen-coh1}) completes the proof of the theorem and
$\overline{\sC_r}(n-m+1,1)$ is given by, via
Lemma~\ref{lem:average-subentropy},
\begin{eqnarray}
\label{eq:gen-coh} \overline{\sC_r}(n-m+1,1) &=& \frac{m-1}{2n} +
\Br{\cI^Q_m(n-m+1,1)-(1+H_{mn}-H_m-H_n)} = \frac{m-1}{2n}.
\end{eqnarray}
Similarly,
\begin{eqnarray}
\label{eq:HS} \overline{\sC_r}(1,1) = \frac{m-1}{2m}.
\end{eqnarray}

\subsection{The proof of Theorem \ref{th:con-ent} of the main
text}\label{subsect:theorem-proof-concentration}

To prove Theorem~\ref{th:con-ent} of the main text, we use the
concentration of measure phenomenon and in particular, L\'evy's
lemma \cite{Ledoux2005, Hayden2006}, which can be stated as follows:

\smallskip
\noindent {\bf L\'evy's Lemma:} Let $\mathcal{F} : \mathbb{S}^{k} \rightarrow \mathbb{R}$ be a Lipschitz function from $k$-sphere to real line with the Lipschitz constant $\eta$ (with respect to the Euclidean norm) and a point $X \in \mathbb{S}^{k}$ be chosen uniformly at random. Then, for all $\epsilon>0$,
\begin{align}
\label{Levy-lemma} \mathrm{Pr} \left\{|\mathcal{F}(X) - \mathbb{E}
\mathcal{F}|  > \epsilon \right\} \leqslant 2\exp\left( -\frac{(k+1)
\epsilon^2}{ 9\pi^3 \eta^2\ln 2 } \right).
\end{align}
Here $\mathbb{E}(\mathcal{F})$ is the mean value of $\mathcal{F}$.
But before we present the proof we need to find the Lipschitz
constant for the relevant function on $\mathbb{S}^{k}$ which is
$G:\mathbb{S}^{nm}\mapsto \mathbb{R}$, defined as
$G\left(\ket{\psi_{AB}}\right) = S\left(\rho_A^{(\mathrm{d})}\right)
- S\left(\rho_{A}\right) = \sC_r(\rho_A)$ where
$\rho_A^{\mathrm{(d)}}$ is the diagonal part of
$\rho_A=\mathrm{Tr}_B(\proj{\psi_{AB}})$.

\begin{lem}
The function $\tilde{F}:\mathbb{S}^{mn}\mapsto \mathbb{R}$, defined
as $\tilde{F}(\ket{\psi_{AB}})=S(\rho^A)$ where
$\rho_A=\mathrm{Tr}_B(\proj{\psi_{AB}})$ and $S$ is the von Neumann
entropy, is a Lipschitz continuous function with Lipschitz constant
$\sqrt{8}\ln m$.
\end{lem}
\begin{proof}
 The proof is given in Ref. \cite{Hayden2006}.
\end{proof}

\begin{lem}
The function $F:\mathbb{S}^{nm}\mapsto \mathbb{R}$, defined as
$F(\ket{\psi_{AB}})=S(\rho_A^{\mathrm{(d)}})$ where
$\rho^{A\mathrm{(d)}}$ is the diagonal part of
$\rho_A=\mathrm{Tr}_B(\proj{\psi_{AB}})$ and $S$ is the von Neumann
entropy, is a Lipschitz continuous function with Lipschitz constant
$\sqrt{8}\ln m$.
\end{lem}
\begin{proof}
We follow the proof strategy of Ref. \cite{Hayden2006}. Let
$\ket{\psi_{AB}} = \sum_{i=1}^{m} \sum_{j=1}^{n} \psi_{ij}
\ket{ij}^{AB}$ and therefore, $\rho_A^{\mathrm{(d)}} =
\sum_{i=1}^{d}p_i \ket{i}\bra{i}$ with
$p_i=\sum_{j}\abs{\psi_{ij}}^2$. Now, $ F(\psi_{AB}) =
-\sum_{i=1}^{m} p_i \ln p_i$. The Lipschitz constant for $F$ can be
bounded as follows:
\begin{align}
\eta^2 := \sup_{\bra{\psi}\psi\rangle\leqslant 1} \nabla F \cdot
\nabla F
&= 4\sum_{i=1}^{m}p_i\left[ 1 + \ln p_i \right]^2\nonumber\\
&\leqslant 4 \left(1+ \sum_{i=1}^{m} p_i(\ln p_i)^2\right)\nonumber\\
&\leqslant 4 \left(1+ (\ln m)^2\right)\leq 8(\ln m)^2,\nonumber
\end{align}
where the last inequality is true for $m\geqslant 3$. Therefore,
$\eta \leqslant \sqrt{8}\ln m$ for $d\geqslant 3$.
\end{proof}
\begin{lem}
The function $G:\mathbb{S}^{nm}\mapsto \mathbb{R}$, defined as
$G\Pa{\ket{\psi_{AB}}} = S\Pa{\rho_A^{(\mathrm{d})}} -
S\Pa{\rho_{A}}$ where $\rho_A^{\mathrm{(d)}}$ is the diagonal part
of $\rho_A=\mathrm{Tr}_B(\proj{\psi_{AB}})$ and $S$ is the von
Neumann entropy, is a Lipschitz continuous function with the
Lipschitz constant $2\sqrt{8}\ln m$.
\end{lem}
\begin{proof}
Take $\sigma_A = \mathrm{Tr}_B\left( \proj{\phi_{AB}} \right)$.
\begin{align*}
 \abs{G\Pa{\ket{\psi_{AB}}}- G\Pa{\ket{\phi_{AB}}}}&:= \abs{S\Pa{\rho_A^{(\mathrm{d})}} - S\Pa{\sigma_A^{(\mathrm{d})}} -\Br{S\Pa{\rho_{A}} - S\Pa{\sigma_{A}}}}\nonumber\\
 & \leqslant \abs{S\Pa{\rho_A^{(\mathrm{d})}} - S\Pa{\sigma_A^{(\mathrm{d})}}}+\abs{S\Pa{\rho_{A}} - S\Pa{\sigma_{A}}}\nonumber\\
 & \leqslant \sqrt{8}\ln m\norm{\ket{\psi_{AB}} - \ket{\phi_{AB}}}_2+\sqrt{8}\ln m\norm{\ket{\psi_{AB}} - \ket{\phi_{AB}}}_2\nonumber\\
 & \leqslant 2\sqrt{8}\ln m\norm{\ket{\psi_{AB}} - \ket{\phi_{AB}}}_2.
\end{align*}
Thus, $G$ is a Lipschitz continuous function with the Lipschitz constant $2\sqrt{8}\ln m$.
\end{proof}
Now applying L\'evy's lemma, Eq. (\ref{Levy-lemma}), to the function
$G\Pa{\ket{\psi_{AB}}}= \sC_r(\rho_A)$, we have
\begin{align}\label{Eq:conc}
\mathrm{Pr} &\Set{\abs{\sC_r(\rho_A)
-\overline{\sC_r}(n-m+1,1)}> \epsilon} \leqslant 2\exp\Pa{-\frac{m n
\epsilon^2}{144 \pi^3 \ln 2 (\ln m)^2}},
\end{align}
for all $\epsilon > 0$. This completes the proof of Theorem \ref{th:con-ent} of the main text.

\end{document}